\documentclass{article}
\usepackage{graphicx}
\usepackage{cite}
\usepackage[left=3cm, right=3cm]{geometry}
\usepackage[latin1]{inputenc}

\usepackage{color}
\usepackage{amsthm,amsmath}
\usepackage{amssymb}
\usepackage{amsfonts}

\newtheorem{theorem}{Theorem}
\newtheorem{definition}{Definition}
\newtheorem{corollary}{Corollary}

\newtheorem{lemma}{Lemma}
\newtheorem{example}{Example}
\newtheorem{remark}{Remark}

\author{Elif Segah Oztas,\\
	\small{Department of Mathematics,}\\
	\small{ Karamanoglu Mehmetbey University, Turkey}\\  \small{esoztas@kmu.edu.tr}}

\begin{document}
\title{Lifted MDS Codes over Finite Fields}
\maketitle
\begin{abstract}

MDS codes are elegant constructions in coding theory and have mode important applications in cryptography, network coding, distributed data storage, communication systems et. In this study, a method is given which MDS codes are lifted to a higher finite field. The presented method satisfies the protection of the distance and creating the MDS code over the $F_q$ by using MDS code over $F_p.$

\end{abstract}
\section{Introduction}

Maximum Distance Separable (MDS) codes \cite{bb} are used across a wide area of modern information technology, cryptography, network coding, data storage etc. In \cite{crp1,crp2,crp3,crp4,crp5,crp6,crp7,crp8,crp9,crp10,crp11,crp12,crp13,crp14},  authors studied on cryptographic approaches. In \cite{nc1,nc3,nc4,nc5,nc6,nc7,nc9,nc10,nc11,nc12,nc13,nc14,nc15,nc16,nc17,nc18,nc19,nc20,nc21,nc22,nc23,nc24,nc25,nc27,nc29,nc30}, MDS codes are studied with network coding.

The main generation method for MDS code is Reed Solomon (RS) codes, especially Generalized Reed Solomon (GRS) codes. In GRS, the code $[n,k,n-k+1]_q$ can obtain where $n \le q$.  
There are some approaches for constructing MDS matrices such that Vandermonde matrix, circulant matrix, Cauchy matrix, Toeplitz matrices etc. \cite{mds1,crp7,crp40,crp41,crp42,crp43,ted}. All of them compute and improve their method over the defined field in the papers. However, calculation complexity increase over the field which has high cardinality for any construction methods for MSD codes, especially in the recursive generating method.

A linear code $C$ with parameters $[n, k, d]_q$ of length $n$ over the finite field $F_q$ where $p$ is a prime and $q$ is a prime power. Any two vectors in $C$ differ in at least $d$
places. Singleton bound is $d \le n-k+1$ and a code satisfying the equality of this bound is called a maximum distance separable (MDS) code.
In this paper, we focus on the extension of existed codes over $F_p$ then we didn't add the MDS code generation method, here. Background on coding theory and related material made be found in \cite{bb}.
In this paper, we give a method to construct MDS codes over $F_q$ ($q=p^t$) by using lift the MDS codes over $F_q$. Moreover, computational complexity is less than other recursive algorithms, and a diversity of the codes are satisfied. These situations give advantages for applications of MDS codes, especially in cryptography

\section{Construction of MDS code over $F_{p^t}$}

In this section, MDS code over $F_p$ ($p$ is a prime) are used to generate MDS codes over $F_{p^t}$ by using distance holder matrix.

\begin{definition}
	Let $M$ is a $n\times n$  diagonal matrix.The entry in the i-th row and j-th column of a matrix $M$ denoted as $m_{ij}$.  $l$ is maximum number of same entries among diagonal entries as follow;
	$$l(M)=\max|\{a : a \in m_{ii} \}|.$$
\end{definition}
For example; $A=\left(\begin{array}{ccc}
	w & 0 &0\\
0 & w^2 &0\\
0 & 0& 1
	
\end{array} \right)
$ and $B=\left(\begin{array}{ccc}
	w & 0 &0\\
	0 & w &0\\
	0 & 0 & 1
\end{array} \right)
$ 
over $F_4$. $l(A)=1$ and $l(B)=2$. In matrix $A$ all diagonal entry has a unique element. In matrix $B$, there are two $w$ and one $1$. Then maximum number of repeated entry in diagonal entries is 2.

\begin{definition}
	Let $M$ be $n\times n$  diagonal matrix  $m_{i,i}\in F_{p^t}^*$ where   $F_{p^t}^*=F_{p^t}-\{0\}$. if $l(M)=1$, it called distance holder matrix. If $l(M)=s$, $M$ is a $s$-distance holder matrix ($s-dh$ matrix).
\end{definition}

The aim of the distance holder matrix is to satisfy the diversity of generators. This means it helps de define different generators for applications of MDS codes. These variations of matrices can be used by the security and multi-node communications systems.

 \begin{lemma}\label{lemmaD}
Let $D$ be an MDS code generation matrix over $F_{p^r}$, then $D'$, obtained by multiplying a row (or column) of $D$ by any element of $F_{p^r}^{*}$, $D$ will also be an MDS code generator matrix.
\end{lemma}

Lemma \ref{lemmaD} is generalized as follows.

\begin{corollary}\label{coro1}
	Let $D$ be an MDS (code generation) matrix, then for any nonsingular diagonal matrices $M_1$ and $M_2$, $M_1 D M_2$ will also be a MDS  matrix.
\end{corollary}

$"\cdot_{F_{p^t}}"$ denote that matrix product operations is over $F_{p^t}$,. 

\begin{definition}
	Let $G$ be a generator matrix of  $C$ that is $[n,k,d]$ $MDS$ code over $F_p$. 
$G\cdot_{F_{p^t}} M$ is generator of $C'$ that is a Lifted MDS code of $C$ over $F_p^t$  
where $M$ is a $n \times n$ $dh$ matrix over $F_{p^t}$ and  $p^t > n$. 
\end{definition}

In the following theorem, codes are lifted to an upper field under some restriction. Then, new codes protect the distance and they are still MDS code over the upper field.

\begin{theorem}\label{teoM}
	Let $C$ be a $[n,k,d]_p$ MDS code over $F_p$. Lifted MDS code of $C$  is $C'$ has parameter $[n,k,d]_{q}$.
\end{theorem}

\begin{proof}
	 Let $C$ be a $[n,k,d]$ MDS code over $F_p$. $d$ has been changed by changing column entries of the generator matrix $G$. In column case, changing the distance is a connected characteristic of $F_p$. Because $p$ is prime, there is no polynomial identification for elements. Then there is no restriction for $d$ except characteristic. Then, operation in field extension to   $F_{p^t}$ that same characteristic as $F_p$ and $p^t > n$  is protect the distance at least $d$.
\end{proof}

By Lemma \ref{lemmaD} and Corollary \ref{coro1}, in same field, the matrix and element operation preserve MDS property. In Theorem \ref{teoM} we satisfy this preservation to the upper finite field by the characteristic of the field.

\begin{remark}{\label{remk1}}
	Diversity of generator matrix are satisfied by using finite field $F_{p^t}$ ($p^t > n$) in Theorem \ref{teoM}. Then number of different generator matrices is $\binom{p^t-1}{n}$. 
\end{remark}

\begin{example}{\label{ex1}}
	Let $G$ be a generator matrix of code $C$  over $F_7$.
	$$G=\left[\begin{array}{cccccccc}
1& 0& 0& 6 &4 &2 &5 &3\\
0 &1 &0 &3 &1 &5 &1 &3\\
0 &0 &1 &3 &5 &2 &4 &6
\end{array}\right]$$
	$C$ is a $[8,3,6]$ MDS code.
	
	Let $M$ be a matrix $dh$-matrix $M=diag(w^{244}, w^{28}, w^{326}, w^{294}, w^{239}, w^{76}, w^{212}, w^{84} )$ over $F_{7^3}.$
	$$G'=G\cdot_{F_{p^t}} M=$$$$=
	\left[\begin{array}{cccccccc}
	    1 &    0  &   0 &w^{221} &w^{223}& w^{288}& w^{253}& w^{239}\\
	    0     &1    & 0 &w^{323} &w^{211} &w^{333} &w^{184} &w^{113}\\
	    0   &  0  &   1 & w^{25} &w^{198} &w^{206}  &   2 &w^{271}	
	\end{array}\right] $$
$G'$ generate a $[8,3,6]$ MDS code over $F_{7^3}$. 

Another example for same code over $F_7$:

Let $M$ be a matrix $dh$-matrix $M=diag( w^{108}, w^{191}, w^{261}, w^{312}, w^{95}, w^{249}, w^{278}, w^{47}  )$ over $F_{7^3}.$

	$$G''=G\cdot_{F_{p^t}} M=$$$$=
\left[\begin{array}{cccccccc}
   1   &  0 &    0&  w^{33}& w^{215} &w^{255}& w^{113}& w^{338}\\
    0  &   1 &    0& w^{178}& w^{246}  &   w&  w^{87}& w^{255}\\
    0  &   0  &   1& w^{108}& w^{119}& w^{102}& w^{245}& w^{299}\\
\end{array}\right] $$
$G''$ generate a $[8,3,6]$ MDS code over $F_{7^3}$.
\end{example}

By Example \ref{ex1}, Distance has been preserved. Moreover, diversity for components of the codes and a MDS code over $F_{7^3} $ have been obtained. By Remark \ref{remk1}, lots of different codes that have the same distance can be generated. This situation has importance in security and communication systems.

\section{Conclusion}

In this paper, we give a method called lifted MDS codes. It satisfies that protection the distance, variation of code components, keep the MDS property in higher finite fields. Moreover, complexity for calculation is less than the previous recursive method which is clear, because there are only matrix multiplications for the generation of new MDS code.

\end{document}